\pdfoutput=1
\documentclass[submission,copyright,creativecommons]{eptcs}
\providecommand{\event}{TARK 2019} %
\usepackage{breakurl}             %
\usepackage{underscore}           %

\usepackage[protrusion=false]{microtype} %

\usepackage{graphicx}
\usepackage{latexsym}

\newcommand{\sigmastar}{\ensuremath{\Sigma^\ast}}
\usepackage{tikzit}
\usepackage{pifont}
\usepackage{booktabs}
\usepackage{amsmath,amsfonts,amssymb,amsthm,nicefrac}
\usepackage{algorithm,algorithmic}
\usepackage{tabularx,multirow}
\usepackage{subfigure,color,epsfig}

\makeatletter%
\newcommand{\singlespacing}{\let\CS=
\@currsize\renewcommand{\baselinestretch}{1}\tiny\CS}
\newcommand{\singlespacingplus}{\let\CS=
\@currsize\renewcommand{\baselinestretch}{1.25}\tiny\CS}
\newcommand{\doublespacing}{\let\CS=
\@currsize\renewcommand{\baselinestretch}{1.75}\tiny\CS}
\newcommand{\extradoublespacing}{\let\CS=
\@currsize\renewcommand{\baselinestretch}{1.9}\tiny\CS}
\newcommand{\nicenicespacing}{\let\CS=
\@currsize\renewcommand{\baselinestretch}{1.9}\tiny\CS}
\newcommand{\draftspacing}{\let\CS=
\@currsize\renewcommand{\baselinestretch}{2.0}\tiny\CS}
\newcommand{\hugedraftspacing}{\let\CS=
\@currsize\renewcommand{\baselinestretch}{2.4}\tiny\CS}
\newcommand{\niceonespacing}{\let\CS=\@currsize\renewcommand{\baselinestretch}{1.1}\tiny\CS}
\newcommand{\nicetwospacing}{\let\CS=\@currsize\renewcommand{\baselinestretch}{1.2}\tiny\CS}
\newcommand{\nicethreespacing}{\let\CS=\@currsize\renewcommand{\baselinestretch}{1.3}\tiny\CS}
\newcommand{\singlespacingplusplus}{\let\CS=\@currsize\renewcommand{\baselinestretch}{1.35}\tiny\CS}
\newcommand{\nicefourspacing}{\let\CS=\@currsize\renewcommand{\baselinestretch}{1.4}\tiny\CS}
\newcommand{\nicefivespacing}{\let\CS=\@currsize\renewcommand{\baselinestretch}{1.5}\tiny\CS}
\newcommand{\nicesixspacing}{\let\CS=\@currsize\renewcommand{\baselinestretch}{1.6}\tiny\CS}
\newcommand{\nicesevenspacing}{\let\CS=\@currsize\renewcommand{\baselinestretch}{1.7}\tiny\CS}
\newcommand{\niceeightspacing}{\let\CS=\@currsize\renewcommand{\baselinestretch}{1.8}\tiny\CS}
\newcommand{\niceninespacing}{\let\CS=\@currsize\renewcommand{\baselinestretch}{1.9}\tiny\CS}
\makeatother%

\newcommand\qedblob{\mbox{\ding{113}}}

\newcommand{\onlinesystemb}[1]{\mathrm{online}\hbox{-}\allowbreak\mathrm{{#1}}\hbox{-}\allowbreak\mathrm{Bribery}}
\newcommand{\onlinesystembk}[2]{\mathrm{online}\hbox{-}\allowbreak\mathrm{{#1}}\hbox{-}\allowbreak\mathrm{Bribery}[#2]}

\newcommand{\onlinesystemdb}[1]{\mathrm{online}\hbox{-}\allowbreak\mathrm{{#1}}\hbox{-}\allowbreak\mathrm{Destructive}\hbox{-}\allowbreak\mathrm{Bribery}}
\newcommand{\onlinesystemdbk}[2]{\mathrm{online}\hbox{-}\allowbreak\mathrm{{#1}}\hbox{-}\allowbreak\mathrm{Destructive}\hbox{-}\allowbreak\mathrm{Bribery}[k]}

\newcommand{\onlinesystemwb}[1]{\mathrm{online}\hbox{-}\allowbreak\mathrm{{#1}}\hbox{-}\allowbreak\mathrm{Weighted}\hbox{-}\allowbreak\mathrm{Bribery}}
\newcommand{\onlinesystemwbk}[2]{\mathrm{online}\hbox{-}\allowbreak\mathrm{{#1}}\hbox{-}\allowbreak\mathrm{Weighted}\hbox{-}\allowbreak\mathrm{Bribery}[#2]}

\newcommand{\onlinesystemdwb}[1]{\mathrm{online}\hbox{-}\allowbreak\mathrm{{#1}}\hbox{-}\allowbreak\mathrm{Destructive}\hbox{-}\allowbreak\mathrm{Weighted}\hbox{-}\allowbreak\mathrm{Bribery}}
\newcommand{\onlinesystemdwbk}[2]{\mathrm{online}\hbox{-}\allowbreak\mathrm{{#1}}\hbox{-}\allowbreak\mathrm{Destructive}\hbox{-}\allowbreak\mathrm{Weighted}\hbox{-}\allowbreak\mathrm{Bribery}[#2]}

\newcommand{\onlinesystempb}[1]{\mathrm{online}\hbox{-}\allowbreak\mathrm{{#1}}\hbox{-}\allowbreak\mathrm{\$Bribery}}
\newcommand{\onlinesystempbk}[2]{\mathrm{online}\hbox{-}\allowbreak\mathrm{{#1}}\hbox{-}\allowbreak\mathrm{\$Bribery}[#2]}

\newcommand{\onlinesystemdpb}[1]{\mathrm{online}\hbox{-}\allowbreak\mathrm{{#1}}\hbox{-}\allowbreak\mathrm{Destructive}\hbox{-}\allowbreak\mathrm{\$Bribery}}
\newcommand{\onlinesystemdpbk}[2]{\mathrm{online}\hbox{-}\allowbreak\mathrm{{#1}}\hbox{-}\allowbreak\mathrm{Destructive}\hbox{-}\allowbreak\mathrm{\$Bribery}[#2]}

\newcommand{\onlinesystempwb}[1]{\mathrm{online}\hbox{-}\allowbreak\mathrm{{#1}}\hbox{-}\allowbreak\mathrm{Weighted}\hbox{-}\allowbreak\mathrm{\$Bribery}}
\newcommand{\onlinesystempwbk}[2]{\mathrm{online}\hbox{-}\allowbreak\mathrm{{#1}}\hbox{-}\allowbreak\mathrm{Weighted}\hbox{-}\allowbreak\mathrm{\$Bribery}[#2]}

\newcommand{\onlinesystemdpwb}[1]{\mathrm{online}\hbox{-}\allowbreak\mathrm{{#1}}\hbox{-}\allowbreak\mathrm{Destructive}\hbox{-}\allowbreak\mathrm{Weighted}\hbox{-}\allowbreak\mathrm{\$Bribery}}
\newcommand{\onlinesystemdpwbk}[2]{\mathrm{online}\hbox{-}\allowbreak\mathrm{{#1}}\hbox{-}\allowbreak\mathrm{Destructive}\hbox{-}\allowbreak\mathrm{Weighted}\hbox{-}\allowbreak\mathrm{\$Bribery}[#2]}

\newcommand{\cale}{{\cal E}}
\newcommand{\calc}{{\cal C}}
\newcommand{\condition}{\mid}
\def\land{{\; \wedge \;}}
\newcommand{\littlep}{{p}}

\newcommand{\sigmatwo}{{\Sigma_2^{\littlep}}}

\newcommand{\deltatwo}{{\Delta_2^{\littlep}}}

  \newtheorem{theorem}{Theorem}[section]
  
  \newtheorem{proposition}[theorem]{Proposition}
  \newtheorem{fact}[theorem]{Fact}
  \newtheorem{definition}[theorem]{Definition}
\newtheorem{observation}[theorem]{Observation}

\newcommand{\manyone}{\ensuremath{\leq_{m}^{p}}}
\newcommand{\dttred}{\ensuremath{\leq_{dtt}^{p}}}

\newcommand{\p}{\ensuremath{\mathrm{P}}}
\newcommand{\np}{\ensuremath{\mathrm{NP}}}
\newcommand{\npnp}{\ensuremath{\np^{\np}}}

\newcommand{\conp}{\ensuremath{\mathrm{coNP}}}

\newcommand{\pspace}{\ensuremath{\mathrm{PSPACE}}}

\newcommand{\OMIT}[1]{} %

\newcommand\sigmalevel[1]{\ensuremath{{\Sigma^p_{#1}}}}

\newcommand\pilevel[1]{\ensuremath{{\Pi^p_{#1}}}}

\def\literalqed{{\ \nolinebreak\hfill\mbox{\qedblob\quad}}}

\newcommand{\systemucm}[1]{{\mathrm{{#1}}\hbox{-}\mathrm{UCM}}}

\newcommand{\systemwcm}[1]{{\mathrm{{#1}}\hbox{-}\mathrm{WCM}}}

\newcommand{\systemducm}[1]{{\mathrm{{#1}}\hbox{-}\mathrm{DUCM}}}

\newcommand{\systemdwcm}[1]{{\mathrm{{#1}}\hbox{-}\mathrm{DWCM}}}

\newcommand{\onlinesystemucm}[1]{{\mathrm{online}\hbox{-}\mathrm{{#1}}\hbox{-}\mathrm{UCM}}}

\newcommand{\onlinesystemwcm}[1]{{\mathrm{online}\hbox{-}\mathrm{{#1}}\hbox{-}\mathrm{WCM}}}

\newcommand{\onlinesystemducm}[1]{{\mathrm{online}\hbox{-}\mathrm{{#1}}\hbox{-}\mathrm{DUCM}}}

\newcommand{\onlinesystemdwcm}[1]{{\mathrm{online}\hbox{-}\mathrm{{#1}}\hbox{-}\mathrm{DWCM}}}

\newcounter{alg}

\newcounter{subalg}

\newcommand{\lahnote}[1]{}

\newcommand{\jrnote}[1]{}

\newcommand{\ehnote}[1]{}

\newcommand{\pair}[1]{\mathopen\langle{#1}\mathclose\rangle}

\usepackage{bigfoot} \DeclareNewFootnote{default}
\title{The Complexity
  of Online Bribery in Sequential Elections (Extended Abstract)}
\author{Edith Hemaspaandra\thanks{Supported in part
    by NSF grant DUE-1819546
    and
    a Renewed Research Stay grant from the
      Alexander von Humboldt Foundation.
Work done
    in part 
while 
  on sabbatical visits to ETH-Z\"urich and the University of
  D\"usseldorf.}
\institute{Department of Computer Science\\
        Rochester Institute of Technology\\
        Rochester, NY 14623, USA}
      \email{eh@cs.rit.edu}
\and
Lane A. Hemaspaandra\thanks{Supported in part by
      a Renewed Research Stay grant from the
      Alexander von Humboldt Foundation.
Work done
    in part 
while 
  on sabbatical visits to ETH-Z\"urich and the University of
  D\"usseldorf.}
\institute{Department of Computer Science\\
        University of Rochester\\
        Rochester, NY 14627, USA}
      \email{http://www.cs.rochester.edu/u/lane}
\and
J{\"o}rg Rothe\thanks{Supported in part by DFG grant RO~1202/14-2.}
\institute{Institut f\"ur Informatik \\
        Heinrich-Heine-Universit{\"a}t D{\"u}sseldorf  \\
        40225 D\"usseldorf, Germany}
        \email{rothe@hhu.de}
      }
\def\titlerunning{The Complexity of Online Bribery in Sequential Elections}
\def\authorrunning{Edith Hemaspaandra, Lane A. Hemaspaandra \& J\"{o}rg Rothe}

            \date{June 14, 2019} %

\newcounter{extremeleftlistcounter}
\newenvironment{extremeleftlistNICE}%
  {\begin{list}{\arabic{extremeleftlistcounter}.~~}{\usecounter{extremeleftlistcounter}%
        \setlength{\itemsep}{-0pt}%
        \setlength{\labelsep}{0pt}\setlength{\leftmargin}{6pt}%
        \setlength{\labelwidth}{6pt}\setlength{\listparindent}{0pt}}}%
  {\end{list}}

\begin{document}
\sloppy

\maketitle

\begin{abstract}  
  Prior work on the complexity of bribery assumes that the bribery
  happens simultaneously, and that the 
  briber
  has full knowledge of all voters' votes.  But neither of those
  assumptions always holds.  In many real-world settings, votes come in
  sequentially, and the briber may have a 
  use-it-or-lose-it moment
  to decide whether to bribe/alter a given
  vote, and at the time of making that decision, the briber may not
  know what votes
  remaining
  voters are planning on casting.

  In this paper, we introduce a model for, and initiate the study of,
  bribery in such an online, sequential setting.  We show that even
  for election systems whose winner-determination problem is
  polynomial-time computable, an online, sequential setting may vastly
  increase the complexity of bribery, in fact jumping the problem up
  to completeness for high levels of the polynomial hierarchy or even
  $\pspace$.
  On the other hand, we show that for some
  natural, important election systems, such a dramatic
  complexity increase does not occur, and we pinpoint the complexity
  of their bribery problems in the online, sequential setting.

\end{abstract}

\section{Introduction}\label{sec:introduction}
In computational social choice theory,
the three most studied types of attacks on elections
are bribery, control, and
manipulation,
and the models of those that are studied seek to model the analogous
real-world actions.

Such studies are typically carried out for the model in which all the voters
vote simultaneously.  That sometimes is the case in the real world.
But it also is sometimes the case that the voters vote in
sequence---in what is sometimes called a roll-call election.

That type of setting has been relatively recently introduced and
studied for control and manipulation---in particular, studies have
been done of both control and manipulation 
in the so-called online, sequential
setting~\cite{hem-hem-rot:j:online-manipulation,hem-hem-rot:j:online-candidate-control,hem-hem-rot:j:online-voter-control}.
In the present paper, we study the complexity of, and algorithms for,
the online, sequential case of bribery.

Briefly put, we are studying the case where the
voting order (and the voter weights and cost of bribing each voter)
is known ahead of time to the briber.  
But at the moment a voter seeks to vote, the voter's planned
vote is revealed to the briber, who then has a use-it-or-lose-it
chance to bribe the voter, by paying the voter's bribe-price
(and doing so allows that vote to be changed to any vote the briber desires).

The problem we are studying is the complexity of that decision.  In
particular, how hard is it to decide whether under optimal play on the part
of the briber there is an action for the
briber towards the current voter 
such that under continued future optimal play by the briber (in the
face of all future revelations of unknown information being
pessimal), the briber can reach a certain goal (e.g., having one of
his or her two favorite candidates win; or not having any of his or
her three most hated candidates win).

What is an example of a situation
that might be modeled by sequential bribery?
One example would be 
the case of a new department chair meeting
sequentially with each faculty member, on the morning of a noon
faculty vote on some
crucial issue, and at each meeting quickly finding the faculty member's
preferences regarding the issue and assessing how much of a
bribe would be needed to change those, and then deciding whether
or not
to
commit that amount from the chair's (limited) discretionary raise
pool in order to execute a bribe of that faculty member.
(See the final three paragraphs
of
Section~3
of the technical-report version of
this paper~\cite{hem-hem-rot:t:online-bribery}
for a more detailed discussion of issues regarding the model,
the varying forms the costs in bribery can take (from actual
dollars to time or effort spent to risk accepted), and the fact
that, despite the typical associations with the word ``bribery,''
in many settings bribery is not modeling immoral or evil acts.)

The following list presents the section structure of our results,
and mentions some of the novel proof approaches needed to obtain them.

\begin{enumerate}

\item Sections~\ref{sec:general-upper-bounds-unlimited}
  and~\ref{sec:general-upper-bounds-limited} establish our upper
  bounds---of $\pspace$ and the $\pilevel{2k+1}$ level of the
  polynomial hierarchy---on online
  bribery (i.e., online, sequential bribery,
  but we will for the rest of the paper and especially in our problem names 
  often omit the word ``sequential'' when the word ``online''
  is present)
  in the general case and in
  the case of being restricted to at most $k$ bribes.

  The $\pilevel{2k+1}$ upper bounds are far less straightforward than
  upper bounds in the polynomial hierarchy typically are.  Since
  bribes can occur on any voter (until one runs out of allowed
  bribes), and so a yes-no decision has to be made, even for the case
  of at most $k$ bribes, there can be long strings of alternating
  existential and universal choices in the natural alternating Turing
  machine programs for the problems.  And so there is the threat that
  one can prove merely a PSPACE upper bound.

  However, in
  Section~\ref{sec:result-about-altern}, we prove a more general
  result about alternating Turing machines that, while perhaps having
  many alternations between existential and universal choices, make
  most of them in a ``boring'' way.  Basically, we show that in the
  relevant setting one can pull much of the existential guessing
  upstream and make it external to the alternating Turing machine, and
  indeed one can do so in such a way that one transforms the problem
  into the disjunction of a polynomial number of uniformly generated
  questions about actions of alternating Turing machines each of which
  itself has at most $2k+1$ alternation blocks.  From that, we
  establish the needed upper bound, both for the relevant abstract
  case of alternating Turing machines and for our online bribery
  problems.

\item\label{p:why-cool}Section~\ref{sec:match-lower-bounds} proves that there are
  election systems with simple winner problem such that each of the
  abovementioned upper-bounds is tight, i.e., that
  PSPACE-completeness holds or $\pilevel{2k+1}$-completeness holds.

  There is
  a substantial, novel challenge that the proof here has to
  overcome.  Namely, to prove for example $\pilevel{2k+1}$-hardness,
  we generally need to reduce from quantified boolean formulas with
  particular quantifiers applying to particular variables.  However,
  in online bribery, the briber is allowed to choose where to do the
  bribing.  This in effect corresponds to having a formula with
  clusters of quantified variables, yet such that we as we 
  attempt to prove theorems related to
  these  structures don't have control over which quantifiers are
  existential and which are universal.  Rather, in effect what the
  online bribery setting will test is whether there exists an
  assignment (consistent with the number of bribes allowed---which
  limits the number of existential quantifiers one can set) of each
  quantifier to be either existential or universal, such that for that
  quantifier-assignment the formula evaluates as true.  (This
  is not at all the same as quantifier exchange; in quantifier exchange, the
  exchanged quantifiers move around together with their associated
  variables.)

  However, we handle this by showing how to construct a new formula
  that builds in protection against this setting.  In particular, we
  note that one can take a quantified boolean formula and turn it into one
  such that, in this Wild West setting of quantifier assignment, the new
  formula can be made true by a legal (i.e., having at most
  as many $\exists$ quantifiers as the original formula) quantifier assignment
  exactly if the original formula is true.

\item{} 
In Section~\ref{sec:online-brib-spec},
we look at the complexity of online bribery
for various natural systems. We show that for both Plurality and
Approval, it holds that priced, weighted online bribery is $\np$-complete,
whereas all other problem variants of online bribery are in~$\p$.
This also shows that bribery can be harder than
online bribery for natural systems.
In addition, we provide complete dichotomy theorems that distinguish
NP-hard from easy cases for
all our online bribery problems for scoring rules
and additionally we show that Veto
elections, even with three candidates, have even higher lower
bounds for weighted online bribery, namely $\p^{\np[1]}$-hardness.
\end{enumerate}

Due to space limits, many of our proofs and much additional
discussion can be found only in the
technical-report version of this paper~\cite{hem-hem-rot:t:online-bribery}.
In particular, only three of our proofs are included in this
conference version and at times this version may refer to proofs
found in the technical report as if they were available to
and known to the reader.

\section{Related Work}\label{s:related}
Our paper's general area is computational social choice, in which 
studying the complexity of election and preference aggregation
problems and manipulative attacks on them is a central theme.
There
are many
excellent surveys and book chapters on computational social
choice~\cite{bra-con-end:b:comsoc,rot:b:econ,bra-con-end-lan-pro:b:handook-of-comsoc},
and computational
social choice and computational complexity have a long history of
close, mutually beneficial interaction (see the survey~\cite{hem:c:bffs}).  

The prior papers most related to our work are the papers that
defined and studied the complexity of 
online control~\cite{hem-hem-rot:j:online-candidate-control,hem-hem-rot:j:online-voter-control},
of online manipulation~\cite{hem-hem-rot:j:online-manipulation}, and
of (regular) bribery~\cite{fal-hem-hem:j:bribery}.  Particularly
important among those is online manipulation, as we will show
connections/inheritance between online manipulation and our model.  We
also will show connections/inheritance between (regular) 
manipulation and our model.  Regular manipulation was introduced by
Bartholdi, Tovey, and Trick~\cite{bar-tov-tri:j:manipulating} in the
unweighted case and by Conitzer, Sandholm, and
Lang~\cite{con-lan-san:j:when-hard-to-manipulate} in the weighted
case.

The existing work most closely related to our work on the effect of
formulas on limits on existential actions is the work on online voter
control~\cite{hem-hem-rot:j:online-voter-control}, though the issues
tackled here are different and harder.

The work of Xia and Conitzer~\cite{con-xia:c:stackelberg-sequential}
(see also
\cite{des-elk:c:sequential-voting,slo:j:sequential-voting,dek-pic:j:sequential-voting-binary-elections})
that
defines and explores the Stackelberg voting game is also about
sequential voting, although unlike this paper
their analysis is game-theoretic
and is about manipulation rather than bribery.
Sequential (and related types of) voting have also been studied
in an axiomatic way~\cite{ten:c:transitive-voting}
and using Markov decision processes~\cite{par-pro:c:dynamic-social-choice},
though neither of those works focuses on issues of bribery.
Poole and Rosenthal~\cite{poo-ros:b:congress} provide a history
of roll-call voting.

Our approach to the briber's goal, which is assuming worst-case
revelations of information, is inspired by the approach used in the
area known as online algorithms~\cite{bor-ely:b:online-algorithms}.

Interesting work that is related---though somewhat distantly---in
flavor to our study is the paper of Chevaleyre et
al.~\cite{che-lan-mau-mon-xia:j:possible-winners-adding-welcome}
on the addition of candidates.  They also focus on the moment at which one
has to make a key decision, in their case whether all of a group of
potential additional candidates should be added.

\section{Preliminaries}\label{sec:preliminaries}

\subsection{Basics}\label{sec:basics}
$\p$ is the class of decision problems in deterministic polynomial time.
$\np$ is the class of decision problems in nondeterministic polynomial time.
For each $k\geq 0$,
$\sigmalevel{k}$
is the class of decision problems in the $k$th $\Sigma$ level of
the polynomial hierarchy~\cite{mey-sto:c:reg-exp-needs-exp-space,sto:j:poly},
e.g., $\sigmalevel{0} = \p$,
$\sigmalevel{1} = \np$,
and $\sigmalevel{2} = \npnp$ (i.e., the class of sets accepted by
nondeterministic polynomial-time 
oracle Turing machines
given unit-cost access to an NP oracle).
For each $k\geq 0$,
$\pilevel{k} = \{ L \condition \overline{L} \in \sigmalevel{k}\}$, e.g., 
$\pilevel{0} = \sigmalevel{0} = \p$,
$\pilevel{1} = \conp$, and 
$\pilevel{3} = \conp^{\np^\np}$.
Further, $\deltatwo = \p^{\np}$ is the class of decision problems
solvable by a deterministic polynomial-time oracle Turing machine with
access to an $\np$ oracle, and $\p^{\np[1]}$ is the same class
restricted to one oracle query.

We say that $A \manyone B$ ($A$ polynomial-time many-one reduces to $B$)
exactly if there is a polynomial-time computable function $f$ such that
$(\forall x)[ x\in A \iff f(x) \in B]$.
\begin{fact}\label{f:m-closure}
For each complexity class $\calc
\in \{
\sigmalevel{0},\allowbreak
\sigmalevel{1},\allowbreak
\pilevel{1},\allowbreak
\sigmalevel{2},\allowbreak
\pilevel{2}, \dots\}$,
$\calc$ is closed downwards under 
polynomial-time many-one reductions, i.e.,
$(B  \in \calc \land A \manyone B) \implies A \in \calc$.
\end{fact}

Each of the classes mentioned in
Fact~\ref{f:m-closure} is even closed
downwards under what is known as polynomial-time
disjunctive truth-table reducibility~\cite{lad-lyn-sel:j:com}.
Disjunctive truth-table reducibility can be defined as follows.
We say that $A \dttred B$ ($A$ polynomial-time disjunctive truth-table
reduces to $B$)
exactly if there is a polynomial-time computable function $f$ such
that, for each $x$, it holds that (a)~$f(x)$ outputs a list of 0 or more
strings, and
(b)~$x \in A$ if and only if at least one string output by $f(x)$ is a
member of $B$.
(Polynomial-time many-one reductions are simply the special case of
polynomial-time disjunctive truth-table reductions where 
the polynomial-time disjunctive truth-table reduction's output-list function
is required to always
contain exactly one element.)
\begin{fact}\label{f:dtt-closure}
For each complexity class $\calc
\in \{
\sigmalevel{0},\allowbreak
\sigmalevel{1},\allowbreak
\pilevel{1},\allowbreak
\sigmalevel{2},\allowbreak
\pilevel{2}, \dots\}$,
$\calc$ is closed downwards under 
polynomial-time disjunctive truth-table reductions, i.e., 
$(B  \in \calc \land A \dttred B) \implies A \in \calc$.
\end{fact}
The above fact is obvious for P, and is
easy to see and well known
for NP and $\conp$ (for example, the results follow
immediately from
the
result of Selman~\cite{sel:j:reductions-pselective} that NP
is closed downwards under so-called positive Turing
reductions).
The
results for the NP and coNP cases relativize (as Selman's mentioned
result's proof clearly relativizes), 
and that
gives
(namely, by relativizing the NP and coNP cases by complete sets
for NP, $\npnp$, etc.)\ the
claims for the higher levels of the polynomial hierarchy
(in fact, it gives something even stronger, since it gives downward
closure under disjunctive truth-table reductions that themselves
are relativized, but we won't need that stronger version in this paper).

Since all the many-one and disjunctive truth-table reductions discussed
in the paper will be polynomial-time ones, we henceforth will
sometimes skip the words ``polynomial-time'' when speaking of a polynomial-time
many-one or disjunctive truth-table reduction.

A set $L$ is said to be (polynomial-time many-one) hard for a class
$\calc$ (for short, ``$L$ is $\calc$-hard'')
exactly if $(\forall B \in \calc)[B
\manyone L]$.  If in addition $L \in \calc$, we say that 
$L$ is polynomial-time many-one complete for $\calc$, or simply
that $L$ is $\calc$-complete.

An (unweighted) election system $\cale$ takes as input
a voter collection $V$ and a candidate set $C$, such that
each element of $V$ contains a voter name
and a preference order
over the candidates in $C$;
and for us in this paper preference orders are always
total orders,
except when we are speaking of approval voting where
the preference orders are bit-vectors from $\{0,1\}^{\|C\|}$.
(For
  the rest of the preliminaries, we'll always speak of
  total orders as the preference orders' type, with it being implicit
  that when later in the paper we speak of
  and prove results about approval voting, all such places will
  tacitly be viewed as speaking of bit-vectors.)
The election system maps from that to
a (possibly nonproper) subset of $C$, often called the
winner set.
We often will call each element of $V$ a vote,
though as is common
sometimes we will use the term vote to refer just to the preference
order.
We often will use the variable names $\sigma$, $\sigma_1$, $\sigma_2$,~\dots,
$\sigma_i$ for total orders.
We allow election systems to, on some inputs, have no
winners.\footnote{Although
in social choice this is often disallowed, as has been 
discussed previously, see, e.g.,~\cite[Footnote~3]{fit-hem-hem:j:xthenx},
artificially excluding the case of
no winners
is unnatural, and many papers in computational
social choice allow this case.
A typical
real-world motivating example is that in Baseball Hall of Fame votes,
having no inductees in a given year is a natural outcome that
has at times occurred.}

For a given  (unweighted, simultaneous) election system, $\cale$, the
(unweighted) winner (aka the winner-determination) problem (in the
unweighted case) 
is the set $\{\pair{C,V,c} \condition c \in C \land$ $c$
is a winner of the election $(C,V)$ under election system~$\cale\}$.

For a given (unweighted~[sic], simultaneous) election system, $\cale$,
the winner problem in the weighted~[sic] case will be the set of all
strings $\pair{C,V,c}$ such that $C$ is a candidate set, $V$ is a set
of weighted (via binary nonnegative integers as weights) votes (each
consisting of a voter name and a total order over $C$),
$c \in C$, and in the unweighted election created from this by
replacing each $w$-weighted vote in $V$ with $w$~unweighted copies of
that same vote, $c$~is a winner in that election under the
(unweighted) election system~$\cale$.
For an election system
$\cale$, it is clear that if the winner problem in the weighted
case is in $\p$, then so is the winner problem in the unweighted
case.  However, it is not hard to
show that there are election systems $\cale$ for which the
converse fails.
The above approach to defining the weighted winner problem
is natural and appropriate for the election systems discussed
in this paper.
(The technical-report
version of this 
paper~\cite{hem-hem-rot:t:online-bribery} has a 
detailed discussion of the strengths and weaknesses of using this approach
to the weighted winner problem in other settings, see also
the discussion and results in~\cite{fit-hem:t:succinct}.)

\subsection{Online Bribery in Sequential Elections}\label{sec:online-brib-sequ}

This paper is about the study of online bribery in sequential
elections.  In this setting, we are---this is the sequential
part---assuming that the voters vote in a well-known order,
sequentially, with each casting a ballot that expresses preferences
over all the candidates.  And we are assuming---this is the online
part---that the attacker, called ``the briber,'' as each new vote comes
in has his or her one and only chance to bribe that voter, i.e., to
alter that vote to any vote of the briber's choice.

Bribery has aspects of both the other standard
types of electoral attacks: bribery is like
manipulation in that one changes votes and it is like
(voter) control in that one is deciding on a set of voters (in
the case of bribery, which ones to bribe).  Reflecting
this, our model follows as closely as possible the relevant parts of
the existing models that study manipulation and control in online
settings~\cite{hem-hem-rot:j:online-manipulation,hem-hem-rot:j:online-voter-control,hem-hem-rot:j:online-candidate-control}.
In particular, we will follow insofar as possible both the model
of, and the notation of the model of, the
paper by Hemaspaandra, Hemaspaandra, and
Rothe~\cite{hem-hem-rot:j:online-voter-control} that
introduced the
study of online voter control in sequential elections. In
particular, we will follow the flavor of their model of control by
deleting voters, except here the key decision is not
whether to delete a given voter, but rather is whether a given voter
should be bribed, i.e., whether the voter's vote should be erased and
replaced with a vote supplied by the briber.  That ``replace[ment]''
part is more
similar to what happens in the study of online manipulation, which
was modeled and studied by Hemaspaandra, Hemaspaandra, and
Rothe~\cite{hem-hem-rot:j:online-manipulation}.  We will, as both
those papers do, focus on a key moment---a moment of
decision---and in particular on the complexity of deciding whether
there exists an action the briber can take, at that moment, such that
doing so will ensure, even under the most hostile of conditions regarding
the information that has not yet been revealed, that the briber 
will be able to
meet his or her goal.

If $u$ is a
voter and $C$ is a candidate set, an
\emph{election snapshot for $C$ and $u$} 
is specified by a triple $V =
(V_{<u}, u, V_{>u})$,
which loosely put is 
made up of all voters in the order they vote, each accompanied
in models where there are prices and/or weights with their prices
and/or weights (which in this paper are assumed to
be nonnegative integers coded
in
binary).\footnote{Why do we feel it
natural in most situations
for the prices
and weights be in binary rather than unary?  A TARK referee,
for example, asked whether it was not natural to assume that 
prices and weights would always be small, or if not, would always be multiples
of some integer that when divided out would make the remaining numbers small.

Our
answer is that both prices and weights in many settings tend to be large,
and without any large, shared-by-all divisor.  To see this
clearly regarding weights, 
consider for example  
the number of shares of stock the various stockholders
hold in some large corporation or the
number of residents in each of the states of a country.  Prices 
too are potentially as rich and varied as are individuals and
objects, e.g., in some settings
each person's bribe-price might be the
exact fair market value of his or her house,
or might be closely related to the number of visits a web site they
own has had in the past year.

Pulling back, we note that requiring prices and weights to be in
unary is often tremendously (and arguably inappropriately)
\emph{helping} the algorithms as to
what their complexity is, since in effect one is ``padding'' the
many inputs' sizes as much as exponentially.  But if 
weighted votes are viewed as indivisible objects---and that is
indeed how they are typically
treated in the literature---the right approach
indeed is to code the weights in binary, and not to give
algorithm designers the potentially vastly lowered bar created by
the padding effect of coding the weights in unary.  Indeed,
it is known
in the study of (nonsequential) bribery that changing prices or 
weights to unary can shift problems' complexities from
NP-hardness
to being in deterministic polynomial time~\cite[pp.~500--504]{fal-hem-hem:j:bribery}.

Also, people typically do code natural numbers
in binary, not unary.}
In addition, for each
voter voting before $u$ (namely, the voters in $V_{<u}$), also
included in this listing 
will be
the vote they cast (or if they were bribed, what vote was
cast for them by the briber) 
and whether they were bribed; and for $u$ the listing will also
include the
vote $u$ will cast unless bribed to cast a different vote.
So $V_{>u}$ is simply a list, in the order
they will vote, of the voters, if any, who come
after $u$, each also including the voter's price and/or weight data if we are
in a priced and/or weighted setting.  Further, the vote for $u$ and 
all the votes
in $V_{<u}$ must be votes over the candidate set $C$ (and in particular,
in this paper votes are total orderings of the candidates, e.g.,
$a > b> c$).\footnote{%
There is a slight overloading of notation above, in that we have not
explicitly listed in the structure the location of the mentioned extra
data.  In fact, our actual definition is that the first and last 
components of the 3-tuple $V$ are lists of tuples, and the middle
component is a single tuple.
Each of these contain the appropriate information, as mentioned above.
For example,  for priced, weighted bribery:
\begin{extremeleftlistNICE}
\item\label{item:1}
  the elements of the list $V_{<u}$ will be 5-tuples $(v_i,p_i,w_i,\sigma_i,b_i)$
whose components respectively are the voter's name, the voter's price,
the voter's weight, the voter's cast ballot (which is the voter's
original preference order if the voter was not bribed and is whatever
the voter was bribed into casting if the voter was bribed),
and
a bit specifying whether that vote resulted from being bribed,
\item\label{item:2} the middle component of
$V$ will be a tuple that contains the first four of those five components,
and
\item\label{item:3} the elements of the list $V_{>u}$ will contain the first three
  of the above-mentioned five components.
\end{extremeleftlistNICE}
Similarly, for example, for unpriced, unweighted bribery, the
three tuple types 
would respectively have three components, two components, and one component.

As a remaining tidbit of notational overloading, 
in some places we will speak of $u$ when we
in fact mean the voter name that is the first component of the
tuple that makes up the middle tuple of $V$.  That is, we will use $u$
both for a tuple that names $u$ and gives some of its properties,
and as a stand-in for the voter him- or herself.  Which
use we mean will 
always~be~clear~from~context.%
}

Let us, with the above in hand, define our notions of online bribery
for sequential elections.  Settings can independently allow or not
allow prices and weights, and so we have four basic types of bribery
in our online, sequential model, each having both
constructive and destructive versions.

In the original paper on
nonsequential bribery there were other types of bribery, e.g.,
microbribery,
unary-coding,
and succinct variants~\cite{fal-hem-hem:j:bribery}. 
Many other types have been studied since, e.g., nonuniform
bribery~\cite{fal:c:nonuniform-bribery}  and 
swap- (and its special case \mbox{shift-)~bribery}~\cite{elk-fal-sli:c:swap-bribery}
(see also~\cite{elk-fal:c:shift-bribery,bre-che-fal-nic-nie:j:prices-matter-shift-bribery}).
However,
for compactness and
since they are very natural, this paper focuses completely on
(standard) bribery in its eight
typical versions (as to prices,
weights, and constructive/destructive), except now in
an online, sequential setting.

Our specification of
these problems as languages is
centered around what Hemaspaandra, Hemaspaandra,
and Rothe~\cite{hem-hem-rot:j:online-manipulation} called a
\emph{magnifying-glass moment}.  This is a moment of decision as to
a particular voter.  To capture precisely what information the
briber does and does not have at that moment, and to thus allow us to
define our problems, we define a structure that we will call an
\emph{OBS}, which stands for \emph{online bribery setting}.  An
OBS is defined as a 5-tuple
$(C, V, \sigma, d, k)$, where $C$ is a set
of candidates;
$V = (V_{<u}, u,
V_{u<})$ is an election snapshot for $C$ and~$u$ as discussed
earlier; $\sigma$ is the
preference order of the briber;
$d \in C$ is a distinguished candidate;
and $k$ is a nonnegative integer (representing for unpriced cases the
maximum number of voters that can be bribed, and for priced cases
the maximum total cost, i.e., the sum of the prices 
of all the bribed voters).

Given an election system~$\cale$, we define the \emph{online 
  unpriced, unweighted
  bribery problem}, abbreviated by
$\onlinesystemb{\cale}$, as the following decision problem.
The input is an OBS\@.
And the question is:
Does there exist a legal (i.e., not violating 
whatever bribe limit holds) 
choice by the briber
on
whether to bribe $u$ (recall that $u$ is specified
in the OBS,
namely, via 
the middle component of $V$) and, if the choice is
to bribe, of what vote to bribe $u$ into casting,
such that if the briber makes that choice 
then no matter what votes the remaining voters after
$u$ are (later) revealed
to have, the briber's goal (the meeting of which itself depends on
$\cale$ and will be defined explicitly two paragraphs from
now) can be reached by the current
decision regarding $u$ and by using the briber's future (legal-only, of
course) decisions 
(if any), each being made using the briber's then-in-hand
knowledge about what votes have been cast
by then?

Note that this approach is about alternating quantifiers.  It is
asking whether there is a current choice by the briber such that for all
potential revealed vote values for the next voter there exists a
choice by the briber such that for all potential revealed vote values
for the next-still voter there exists a choice by the briber such
that\ldots~and so on\ldots~such that the resulting winner set
under election system $\cale$ meets the briber's goal.
This is a bit more subtle than it might at first seem.  The briber is
acting very powerfully, since the briber is represented by existential
quantifiers.  But the briber is not all-powerful in this model.  In
particular, the briber can't see and act on future revelations of vote values;
after all, those are handled by a universal quantifier that occurs
downstream from an existential quantifier that commits the briber to
a particular choice.

In the above we have not defined what ``the briber's goal'' is, so
let us do that now.
$W_{\cale}(C,U)$
will denote
the winner set, according
to election system $\cale$,
of the (standard, nonsequential) election $(C,U)$, where $C$ is
the candidate set and $U$ is the set of votes.
By \emph{the briber's goal} we mean, in the constructive case, that if
at the end of the above process $U'$ is the set of votes (some may
be the
original ones and some may be the result of bribes), it holds that
$W_{\cale}(C,U') \cap \{c \condition c \geq_{\sigma} d\} \neq
\emptyset$, i.e., the winner set includes some candidate (possibly
itself being~$d$) that the briber likes at least as much as the briber
likes $d$.  In the destructive case, the goal is to ensure that
no 
candidate that the briber hates as much or more than the briber hates
$d$ belongs to the winner set, i.e., the briber's goal is to ensure
that
$W_{\cale}(C,U') \cap \{c \condition d \geq_{\sigma} c\} =
\emptyset$.

Above, we have defined both 
$\onlinesystemb{\cale}$ and 
$\onlinesystemdb{\cale}$.  Those both are in the
unpriced, 
unweighted 
setting.  And so
as per our definitions, the voters passed in as part of the problem statement
do not come with or need price or weight information.

But our definitions note that for the priced and/or weighted settings,
the OBS will carry the prices and/or weights.  And so the same
definition text that was used above defines all the other cases, except
that one must keep in mind
for the priced cases that when the ``bribery limit'' is mentioned one must
instead speak of the ``bribery budget,'' and in the weighted cases the
winner set $W$ is of course defined in terms of the weighted version
of the given voting system (which must, for that to be meaningful,
have a well-defined notion of what its weighted version is; 
Section~\ref{sec:basics} provides that notion for all systems 
in this paper).
Thus, we
also have tacitly defined the six problems
$\onlinesystempb{\cale}$,
$\onlinesystemdpb{\cale}$, $\onlinesystemwb{\cale}$,
$\onlinesystemdwb{\cale}$, $\onlinesystempwb{\cale}$, and
$\onlinesystemdpwb{\cale}$.
For an unpriced online bribery problem, we will postpend
the problem name with a ``[$k$]'' to define the version where
as part of the problem definition itself the bribery limit is---in
contrast with the above unpriced problem---not part of
the input
but rather is fixed to be the value $k$.
For example, $\onlinesystembk{\cale}{k}$ denotes the
unpriced, unweighted bribery problem where the number of voters who can
be bribed is set not by the problem input but rather is limited to be
at most $k$.  Note that in each of the  ``[$k$]'' variants,
we tacitly are altering the definition of OBS from its standard
5-tuple,
$(C, V, \sigma, d, k)$,
to instead the 4-tuple
$(C, V, \sigma, d)$; that is because for these cases, the $k$ is
fixed as part of the general problem itself, rather than being a
variable part of
the individual instances.  For priced ``[$k$]'' variants,
there will be both a limit (being a variable part
of the input) on the total price of the bribes and a fixed as part
of the general problem itself
limit on the number of voters who can be bribed.

Of course, there are some immediate relationships that hold between
these eight problems.  One has to be slightly careful since there is a
technical hitch here.  We cannot for example simply claim that
$\onlinesystembk{\cale}{k}$ is a subcase of
$\onlinesystempbk{\cale}{k}$.  If we had implemented the unpriced case
by still including prices in the input but requiring them all to be 1,
then it would be a subcase.  But regarding both prices (weights), our
definitions simply omit them completely from problems that are not
about prices (weights).  In spirit, it is a subcase, but formally it
is not.  Nonetheless, we can still reflect the relationship between these
problems, namely, by stating how they are related via polynomial-time
many-one reductions. (We could even make claims regarding more restrictive
reduction types, but since this paper is concerned with complexity classes
that are
closed downwards under polynomial-time many-one
reductions, there is no
reason to do so.)  The following proposition (and the connections that follow from
it by the transitivity of polynomial-time many-one reductions) captures
this.

\begin{proposition}\label{p:reductions}
  \begin{enumerate}
    \item
For each $k>0$ and for each election system $\cale$,
\begin{itemize}
   \item 
      $\onlinesystembk{\cale}{k} \manyone \onlinesystempbk{\cale}{k}$,
    \item 
      $\onlinesystembk{\cale}{k} \manyone \onlinesystemwbk{\cale}{k}$,
\item $\onlinesystempbk{\cale}{k} \manyone\onlinesystempwbk{\cale}{k}$, and 
\item $\onlinesystemwbk{\cale}{k} \manyone\onlinesystempwbk{\cale}{k}$.
\end{itemize}
\item The above item also holds for the case when all of its
  problems are changed to their destructive versions.
\item The above two items also hold for the case when
  all the ``$[k]$''s are removed (e.g., we have 
  $\onlinesystemb{\cale} \manyone \onlinesystempb{\cale}$).
\end{enumerate}
\end{proposition}

As mentioned in the introduction, the technical-report version of
this paper~\cite{hem-hem-rot:t:online-bribery}
includes some additional discussion of
the model, of the flexibility of bribery, and of the fact
that bribery is not necessarily always modeling
immoral/evil acts.

\section{General Upper Bounds and Matching Lower Bounds}\label{sec:general-upper-bounds}
Even for election systems with simple winner problems, the best
general upper bounds that we can prove for our problems reflect an
extremely high level of complexity.

One might
wonder whether
that merely is a weakness in our upper-bound
proofs.  However, in each case, we provide a matching completeness
result
proving that these really are the hardest problems in the
classes their upper bounds put them in.

However, in Section~\ref{sec:online-brib-spec}, we will see that for
many specific natural, important systems, the complexity is
tremendously lower than the upper bounds, despite the fact that the
present section shows that there exist systems that meet the upper
bounds.

\subsection{The General Upper Bound, Without Limits on the Number of Bribes}\label{sec:general-upper-bounds-unlimited}
This section covers upper bounds for
the case when any bribe limit/bribery budget is
passed in through the input---not hardwired into the problem
itself.
\begin{theorem}\label{t:pspace-general}
    \begin{enumerate}
  \item
    For each election system $\cale$ whose winner problem in the unweighted
  case is in
  polynomial time (or even in polynomial space),
each of the problems
  $\onlinesystemb{\cale}$,
  $\onlinesystemdb{\cale}$,
$\onlinesystempb{\cale}$, and 
$\onlinesystemdpb{\cale}$
is in $\pspace$.
\item
      For each election system $\cale$ whose winner problem in the weighted
  case is in
  polynomial time  (or even in polynomial space), each of the problems
$\onlinesystemwb{\cale}$,
$\onlinesystemdwb{\cale}$,
$\onlinesystempwb{\cale}$, and
$\onlinesystemdpwb{\cale}$
is in $\pspace$.
\end{enumerate}
\end{theorem}  
The proof is made available through 
the technical-report version of this paper~\cite{hem-hem-rot:t:online-bribery}.

\subsection{The General Upper Bound, With Limits on the Number of Bribes}\label{sec:general-upper-bounds-limited}

Turning to the case where in the problem the number of bribes has
a fixed bound of $k$, these problems fall into the $\pilevel{2k+1}$
level of the polynomial hierarchy.  That is not immediately obvious.
After all, even when one can bribe at most $k$ times, one still for
each of the current and future voters seems to need to explore the
one-bit-per-voter decision of whether to bribe the voters (plus in
those cases where one does decide to bribe, one potentially has to
explore the exponential---in the number of candidates---possible votes
to which the voter can be bribed).  On its surface, for our problems,
that would seem to say that the number of alternations between
universal and existential moves that the natural polynomial-time
alternating Turing program for our problem would have to make is about
the number of voters---a bound that would not leave the problem in any
fixed level of the polynomial hierarchy, but would merely seem to put
the problem in $\pspace$.

So these problems are cases where even obtaining the stated upper
bound is interesting and requires a twist to prove.  The twist is as follows.
On the surface the exploration of these problems has an unbounded
number of alternations between universal and existential states in
the natural, brute-force alternating Turing machine program.  But
for all but $k$ of the existential guesses on each accepting
path, the guess is a boring one,
namely, we guess that regarding that voter we don't bribe.
We will show, by proving a more general result about alternating
Turing machines and restrictions on the structure of their maximal
existential move segments along accepting paths, a 
$\pilevel{2k+1}$ upper bound on our sets of interest.
In some sense, in terms of being charged
as to levels of the polynomial hierarchy, we will be showing
that if for
a certain collection of 0-or-1 existential
decisions one on each accepting path chooses 0 all but a fixed
number of times (although for the other times one
may then make many more nondeterministic
choices), one can manage to in effect
not be charged at all for the guessing acts
that guessed~0.

We know of only one result in the literature that is anything like
this.  That result, which also came up in the complexity of online
attacks on elections, is a result of Hemaspaandra, Hemaspaandra, and
Rothe~\cite{hem-hem-rot:j:online-voter-control}, where in the context
not of bribery but of voter control they showed that for
each fixed $k>0$ it holds that, for each polynomial-time alternating Turing
machine $M$ whose alternation blocks are each one bit long and that for at
most $k$ of the existential blocks guess a zero, the language accepted
by $M$ is in $\conp$.

In contrast, in the present
paper's case we are in a far more complicated situation, since in bribery our
existential blocks are burdened not just by 1-bit bribe-or-not
decisions, but for the cases when we decide to try bribing, we need to
existentially guess what bribe to do.  And so we do not stay in
$\conp$ regardless of how large
 $k$ is, as held in that earlier case.
 But we show that we can at least limit the growth to at most $2k+1$
 alternating quantifiers---in particular, to the class $\pilevel{2k+1}$.
And since we
later provide problems of this sort that are complete for
$\pilevel{2k+1}$, our $2k+1$ is optimal unless the polynomial
hierarchy collapses.

We will approach this in two steps.  First,
as Section~\ref{sec:result-about-altern},
we will prove the result
about alternating Turing machines.  And then,
as Section~\ref{sec:upper-bound-results},
we will apply that to
online bribery in the case of
only globally fixed numbers of bribers being allowed.

\subsubsection{A Result about Alternating Turing Machines}\label{sec:result-about-altern}

Briefly put, an alternating Turing machine~\cite{cha-koz-sto:j:alternation}
(aka~ATM) is a generalization of
nondeterministic and conondeterministic computation.
We will now briefly review the basics
(see~\cite{cha-koz-sto:j:alternation} for a more complete
treatment).
An ATM can make both universal and existential choices.  For a
universal ``node'' of the machine's action to evaluate to true, all
its child nodes (one each for each of its possible choices) must
evaluate to true.  For an existential ``node'' of the machine's action
to evaluate to true, at least one of its child nodes (it has one child
node for each of its possible choices) must evaluate to true.  A
leaf of the computation tree (a path, at its end) is said to
evaluate to true if the path halted in an accepting state and is
said to evaluate to false if the path halted in a nonaccepting state.
(As our machines are time-bounded, all paths halt.)  Without loss of
generality,
in this paper we assume that each universal or existential node has
either two children (namely, does a universal or existential split
over the choices 0 and 1; we will often call this a 
``1-bit move'') or has exactly one child (it does a
trivial/degenerate
universal or existential choice of an element from the one-element set
$\{0\}$; we will often call this a ``0-bit move'').  
The latter case is in effect a deterministic move, except
allowing degenerate $\forall$ steps of that sort will let us put a
``separator'' between otherwise contiguous $\exists$ computation
segments.
Of course, long existential guesses can be done in this model,
for example by guessing a number of bits sequentially.
An
ATM accepts or rejects based on what its root
node evaluates to (which is determined inductively in the way described
above).

\begin{definition}\label{d:weight}
  Consider a path $\rho$ in the tree of an ATM\@.  The \emph{weight}
  of that path is as follows.  Consider all maximal segments of
  existential nodes along the path.  (As mentioned above, we may
  without loss of generality, and do, assume that each nonleaf node is
  $\exists$ or $\forall$, although perhaps a degenerate such node
  in the way mentioned above).
  The weight of path $\rho$
  is its number of maximal existential segments such that the concatenation
  of the bits guessed in that segment is not the 1-bit string 0.
\end{definition}

Figure~\ref{f:left}
gives an illustration of this.  In the figure,
the illustrated path (the leftmost one at the left edge of the tree)
has weight 0; it has three maximal existential segments, but each
is of length one and makes the guess 0.
\begin{figure}[!t]
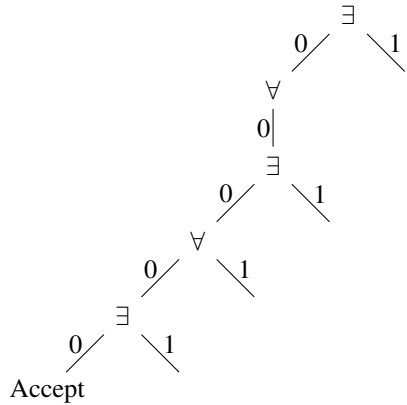

  \small
   \ctikzfig{fig-left-path}
      \caption{\label{f:left}A weight 0 path in the tree of an
        ATM\@.}
\end{figure}

With this definition in hand, we can now state our key theorem showing
that limited weight for ATMs simplifies the complexity of the
languages accepted.  The result is one where one may go back and forth
between thinking it is obvious and thinking it is not obvious.  In
particular, note that even on accepting paths, which must be of weight
at most $k$, it is completely possible that the number of alternations
between existential and universal nodes may be far greater than $k$
and may be far greater than $2k+1$, and indeed may grow unboundedly as
the input's size increases.  What the theorem below is saying is that
despite that, machines with bounded weight on their accepting paths
still accept only $\pilevel{2k+1}$ sets.  

\begin{theorem}\label{t:2k+1}
  Let $k\geq 0$ be fixed.  Each polynomial-time ATM $M$ such that
  on no input does 
  $M$ have an accepting path of weight strictly
  greater than $k$
  accepts a
  language in $\pilevel{2k+1}$.
\end{theorem}

\begin{proof}
  Let $k\geq 0$ be fixed.  Let $L$ be the language accepted by
  polynomial-time ATM $M$ that has the property that each of its
  accepting paths has weight at most $k$.  Our goal is to prove that
  $L \in \pilevel{2k+1}$.  

  We will do so by proving that there is a set $G \in \pilevel{2k+1}$
  such that $L \dttred G$, i.e., $L$ polynomial-time disjunctively
  truth-table reduces to $G$.  By Fact~\ref{f:dtt-closure},
  it follows that $L \in \pilevel{2k+1}$.

  Let $q$ be a nondecreasing polynomial that upper-bounds the running time
  of $M$.  Let $\pair{\cdot,\cdot}$ be a standard pairing function, i.e.,
  a polynomial-time computable, polynomial-time invertible bijection between
  $\sigmastar \times \sigmastar$ and $\sigmastar$.
  Recall that every step of our ATM $M$ involves either a 0-bit existential
  move
  (which we'll think of basically as existentially choosing 0 from the
  choice palette set $\{0\}$)
  or a 1-bit existential move (which, recall, involves choosing
  one element from the choice palette set $\{0,1\}$ with the machine
  enforcing an ``or'' over the two children thus reached)
   or a 0-bit universal move
  (which we'll think of basically as universally choosing 0 from the
  choice palette set $\{0\}$)
  or a 1-bit universal move.  And the 1-bit moves involve successor
  states hinged on whether the move-choice is a 0 or a 1.
 (As
  Turing machines are standardly defined, there can be (one or multiple)
  successor states to a given state, hinged on a (degenerate or
  nondegenerate) choice.)

  $G$ will be the set of all $\pair{x,s}$ such that all of the 
  list of conditions that we will give below hold relative to $x$ and $s$.
  The intuition here is that $s$ is a
  bit-vector whose $i$th bit controls how the $i$th maximal
  existential segment is handled: If that bit is a~1, the segment moves
  forward unrestrained, but if that bit is a~0, then we expect and
  require (and cut off that part of the tree otherwise) the maximal
  existential segment to be a single existential step (either guessing
  a bit from $\{0,1\}$ or the allowed but superfluous existential step
  of guessing a bit from the one-element set $\{0\}$) and we basically
  cut that step out of the tree by replacing it by a trivially 
  universal step.
  Returning to our defining of $G$,
  the set $G$ will be all $\pair{x,s}$ such that all 
  the following claims hold.
    \begin{enumerate}
  \item $x \in \sigmastar$ and $s \in \{0,1\}^{q(|x|)}$.

  \item The number of ``1''s in the bit-string $s$ is at most $k$.
  \item $M$ accepts when we simulate it on input $x$ but with the
    following changes in the machine's action.  

    As one simulates $M$ on a given path, consider the first
    existential node (if any) that one encounters.

    If the first bit of $s$ is
    a 1, then for that node we will directly simulate it, and
    on all paths that follow from this one, on all the following
    existential nodes (if any) that are in an unbroken segment of
    existential nodes from this one, we will similarly directly
    simulate them.  On the other hand, if the first bit of $s$ is a~0,
    then (a)~if it is the case that if the current node
    makes the choice 0 then the node that follows it is an existential
    node, then the current path halts and rejects (because something that
    $s$ is specifying as being
    a maximal existential segment consisting of a single~0
    clearly is not); and (b)~if~(a) does not hold (and so the
    node that follows if we make the choice~0 is either
    universal or a leaf), then do not take
    an existential action at the current node but rather implement it
    as a degenerate universal step (namely, a ``$\forall$'' guess over
    one option, namely, 0, matching as to next state and so on
    whatever the existential node would have done on the choice
    of~0).

    If the path we are simulating didn't already end or get cut off
    during the above-described handling of its first, if any,
    maximal existential segment, then continue on until we hit the
    start of its second existential segment.  We handle that exactly
    as described above, except now our actions are controlled not by
    the first bit of $s$ but by the second.

    And similarly for the third maximal existential segment, the fourth,
    and so on.

    All other aspects of this simulation are unchanged from $M$'s own
    native behavior.
      \end{enumerate}

  Note that $G\in\pilevel{2k+1}$.  Why?  Even in the worst of cases for us,
  the computation of $M$ starts with a $\forall$ block and then has $k$
  $\exists$ blocks each separated by a $\forall$ block; and then we
  finish with a $\forall$ block.  But then
  our ATM as it does that simulation starts with $\forall$
  and has $2k$ alternations of quantifier type, and thus has $2k+1$ alternation
  blocks with $\forall$ as the leading one.  And so by
  Chandra, Kozen, and Stockmeyer's~\cite{cha-koz-sto:j:alternation}
  characterization of the languages accepted by ATMs with that leading
  quantifier and that number of alternations, this set is in
  \pilevel{2k+1}.

  Finally, we argue that $L \dttred G$.  In particular, note that 
  $L \dttred G$ via the reduction that on input $x$ generates
  every length $q(|x|)$ bit-string having less than or equal to $k$
  occurrences of the bit~1, and as its $\dttred$ output list outputs
  each of those paired with $x$.  This list is easily generated
  and is polynomial in size.  In particular, the number of
  pairs in the list is clearly at most
  $\sum_{0\leq j \leq k} {q(|x|) \choose j}
  \leq \sum_{0\leq j \leq k}  (q(|x|))^j \leq (k+1) (q(|x|))^k$.
  That completes the proof.~\end{proof}

In the above, we focused on maximal existential guess sequences, and
limiting the number of those, on accepting paths, whose
bit-sequence-guessed was other than the string 0.  So we barred from
accepting paths any maximal existential guess sequences that contain a
1 and any that have two or more bits.  We mention in passing that we
could have framed things more generally in various ways.  For example,
we could have made each maximal existential guess be of a fixed
polynomial length and could have defined our notion of a ``boring''
guess sequence not as the string ``0'' but as a string of 0's of
exactly that length.  The $\pilevel{2k+1}$ upper bound holds also in
that setting, via only slight modifications to the proof.

\subsubsection{The Upper-Bound Results Obtained by Applying 
  the Previous Section's Result about Alternating Turing Machines%
}\label{sec:upper-bound-results}

\begin{theorem}\label{t:parameterized-case-upper-bounds}
    \begin{enumerate} 
  \item\label{p:k-uw}
For each $k \in \{0,1,2,\dots\}$, and for 
each election system $\cale$ whose winner problem in the unweighted
  case is in
  polynomial time,\footnote{Unlike
    Theorem~\ref{t:pspace-general}, we cannot allow the winner problem
    here to be in $\pspace$ and argue that the rest of the theorem
    holds unchanged.  However, we can allow the winner problem here to
    even be in $\np\cap\conp$, and then the rest of the theorem holds
    unchanged.  The key point to notice to see that that holds is---as
    follows immediately from the
    fact that
    $\np^{\np\cap\conp} = \np$~\cite{sch:j:low}---that
    for each $k \geq 0$, 
    $\np \cap \conp$ is $\pilevel{2k+1}$-low, i.e.,
    that $\left({\pilevel{2k+1}}\right)^{\np\cap\conp} =  \pilevel{2k+1}$.}
 each of the problems
  $\onlinesystembk{\cale}{k}$,
  $\onlinesystemdbk{\cale}{k}$,
$\onlinesystempbk{\cale}{k}$, and 
$\onlinesystemdpbk{\cale}{k}$
is in $\pilevel{2k+1}$.
\item\label{p:k-w}
For each $k \in \{0,1,2,\ldots\}$, and for 
each election system $\cale$ whose winner problem in the weighted
  case is in
  polynomial time,\footnote{As in the case of
    the immediately preceding footnote,
    the rest of the theorem remains unchanged even if we relax
    the ``polynomial time'' to instead be ``$\np\cap\conp$.''}
each of the problems
$\onlinesystemwbk{\cale}{k}$,
$\onlinesystemdwbk{\cale}{k}$,
$\onlinesystempwbk{\cale}{k}$, and
$\onlinesystemdpwbk{\cale}{k}$
is in $\pilevel{2k+1}$.
\end{enumerate}
\end{theorem}  

\begin{proof}
Let $k>0$ be fixed.
Let us start by arguing that 
$\onlinesystempwbk{\cale}{k} \in 
\pilevel{2k+1}$.

As noted (for the case without the bound of~$k$)
in Section~\ref{sec:online-brib-sequ}, what is really going on here
is about alternating quantifiers.  Consider a given input to the problem
$\onlinesystempwbk{\cale}{k}$.
Let the voter under consideration (i.e., $u$) in the focus moment of
that problem
just for
this paragraph be referred to as  $u_1$, and let the ones coming
after it be called, in the order they occur, $u_2$, $u_3$,~$\dots$, $u_\ell$. 
What the membership problem is in essence asking is whether
there \emph{exists} an allowable (within both the price budget and
the global limit of $k$ allowed bribes)
choice as to whether to bribe $u_1$ (and if the decision is to bribe,
then whether there  \emph{exists} a vote to which to bribe $u_1$)
such that, \emph{for each} vote that $u_2$ may then be revealed to have,
there 
\emph{exists} an allowable (within both the price budget and
the global limit of $k$ allowed bribes)
choice as to whether to bribe $u_2$ (and if the decision is to bribe,
then whether there  \emph{exists} a vote to which to bribe $u_2$)
such that, \emph{for each} vote that $u_3$ may then be revealed to
have,~$\dots$,~such that
there \emph{exists} an allowable (within both the price budget and
the global limit of $k$ allowed bribes)
choice as to whether to bribe $u_\ell$ (and if the decision is to bribe,
then whether there  \emph{exists} a vote to which to bribe $u_\ell$)
such that
$W_{\cale}(C,U') \cap \{c \condition c \geq_{\sigma} d\} \neq
\emptyset$ (recall that that inequality says that 
the winner set includes some candidate
that the briber likes at least as much as the briber
likes $d$; $U'$ is here representing the vote set after all
the voting/bribing, as per Section~\ref{sec:online-brib-sequ}'s definitions).

Note that for at most $k$ of the choice blocks associated with
$u_1, u_2, \dots$ can we make the choice to bribe.
 (In fact,
if we have already done bribing of one or more
voters in $V_{< u}$, then our remaining number of allowed
bribes will be less than $k$.)  Keeping that in mind, imagine
implementing the above paragraph's alternating-quantifier-based
algorithm on a polynomial-time ATM\@.  In our model, every step is either
a universal or an existential one, and let us
program up all deterministic computations that are part of the above
via degenerate universal steps.  (We do that rather than
using degenerate existential steps since those degenerate
existential steps would interact fatally with our definition
of maximal existential sequence; we really need those places
where one guesses that one will not bribe to be captured as a
maximal existential sequence of length one with guess bit~0; this
comment is quietly using the fact that when making a 1-bit choice
as to whether to bribe we associate the choice 1 with ``yes bribe
this voter'' and 0 with ``do not bribe this voter.'')
In light of that and the fact that we know that
the weighted winner problem of election system $\cale$ is in~$\p$, the
limit of~$k$
ensures that no accepting path will have weight greater than $k$.
And so by Theorem~\ref{t:2k+1}, we have that
$\onlinesystempwbk{\cale}{k} \in 
\pilevel{2k+1}$.

By the exact same argument,
except changing the test at the end to
$W_{\cale}(C,U') \cap \{c \condition d \geq_{\sigma} c\} =
\emptyset$,
we have that 
$\onlinesystemdpwbk{\cale}{k} \in 
\pilevel{2k+1}$.

From these two results, it follows by Proposition~\ref{p:reductions}
that $\onlinesystemwbk{\cale}{k} \in \pilevel{2k+1}$ and
$\onlinesystemdwbk{\cale}{k} \in \pilevel{2k+1}$.

That completes the proof of part~\ref{p:k-w} of the theorem.
Now, we cannot simply invoke
Proposition~\ref{p:reductions}
to claim that part~\ref{p:k-uw} holds.  The reason is that
part~\ref{p:k-uw}'s hypothesis about the winner problem merely
puts the unweighted winner problem in $\p$, but 
the proof we just gave of part~\ref{p:k-w} used the fact
that for that part we could assume that the
weighted winner problem is in~$\p$.

However, the entire construction of this proof works perfectly
well in the unweighted case, namely, we are only given that
the unweighted winner problem is in $\p$, but the four problems
we are studying are the four unweighted problems of
part~\ref{p:k-uw} of the theorem statement.  So we have that
 each of the problems
  $\onlinesystembk{\cale}{k}$,
  $\onlinesystemdbk{\cale}{k}$,
$\onlinesystempbk{\cale}{k}$, and 
$\onlinesystemdpbk{\cale}{k}$
is in $\pilevel{2k+1}$.  That completes the proof of the theorem.~%
\end{proof}

\subsection{Matching Lower Bounds}\label{sec:match-lower-bounds}
For each of the $\pspace$ and $\pilevel{2k+1}$ upper bounds
established so far in this section, we can 
in fact establish a matching lower bound.  We show
that by, for each, proving that there is
an election system, with a polynomial-time winner
problem, such that the given problem is polynomial-time
many-one hard for the relevant
class (and so, in light of the upper-bound results,
is polynomial-time many-one complete for the relevant class).

\begin{theorem}\label{t:lower}
    \begin{enumerate}
  \item \label{p:lower-unbounded}
    For each problem $I$ from this list of problems:  $\onlinesystemb{\cale}$,
$\onlinesystemdb{\cale}$,
$\onlinesystempb{\cale}$, 
$\onlinesystemdpb{\cale}$,
$\onlinesystemwb{\cale}$,
$\onlinesystemdwb{\cale}$,
$\onlinesystempwb{\cale}$, and
$\onlinesystemdpwb{\cale}$,
there exists an (unweighted)  election system $\cale$, whose
winner problem in both the unweighted case 
and the weighted  case is in
polynomial time,
such that $I$ is $\pspace$-complete.
\item \label{p:lower-2}
For each $k \in \{0,1,2,\ldots\}$, and for 
each problem $I$ from this list of problems:
    $\onlinesystembk{\cale}{k}$,
  $\onlinesystemdbk{\cale}{k}$,
$\onlinesystempbk{\cale}{k}$, 
$\onlinesystemdpbk{\cale}{k}$,
$\onlinesystemwbk{\cale}{k}$,
$\onlinesystemdwbk{\cale}{k}$,
$\onlinesystempwbk{\cale}{k}$, and
$\onlinesystemdpwbk{\cale}{k}$,
there exists an (unweighted)  election system $\cale$, whose
winner problem in both the unweighted case 
and the weighted  case is in
polynomial time,
such that $I$ is $\pilevel{2k+1}$-complete.
\end{enumerate}
\end{theorem}

This is
the most involved and interesting proof in the paper, but
it also is quite long. So 
due to space the proof---and how to extend the proof to ensure
that all the constructed election systems are
simultaneously
(candidate-)neutral
and (voter-)anonymous---is made available through the
technical-report version of this paper~\cite{hem-hem-rot:t:online-bribery}.

\section{Online Bribery for Specific Systems}\label{sec:online-brib-spec}

In this section, we look at the complexity of online bribery
for various natural systems.
For both Plurality and
Approval, we show
that priced, weighted online bribery is
$\np$-complete
but that the election system's other three online bribery variants are in~$\p$.
This also shows that bribery can be harder than
online bribery for natural systems.
In addition, we provide complete dichotomy theorems that distinguish 
NP-hard from easy cases for
all our online bribery problems for scoring rules
and additionally we show that Veto
elections, even with three candidates, have even higher lower
bounds for weighted online bribery, namely $\p^{\np[1]}$-hardness.

The following theorem is useful for proving lower bounds for
online bribery for specific systems.

\begin{theorem}\label{t:m-to-b}
  \begin{enumerate}
\item
(``Regular'') manipulation reduces to corresponding online bribery.
(So, $\systemucm{\cal E}$ reduces to $\onlinesystemb{\cal E}$,
$\systemducm{\cal E}$ reduces to $\onlinesystemdb{\cal E}$,
$\systemwcm{\cal E}$ reduces to $\onlinesystemwb{\cal E}$, and
$\systemdwcm{\cal E}$ reduces to $\onlinesystemdwb{\cal E}$.)

\item
Constructive manipulation in the unique winner model
reduces to corresponding online destructive bribery.
(So, $\systemucm{\cal E}$ in the unique winner model
 reduces to $\onlinesystemdb{\cal E}$ and
$\systemwcm{\cal E}$ in the unique winner model
reduces to $\onlinesystemdwb{\cal E}$.)
\item
Online manipulation reduces to corresponding online priced bribery.
(So, $\onlinesystemucm{\cal E}$ reduces to $\onlinesystempb{\cal E}$,
$\onlinesystemducm{\cal E}$ reduces to $\onlinesystemdpb{\cal E}$,
$\onlinesystemwcm{\cal E}$ reduces to $\onlinesystempwb{\cal E}$, and
$\onlinesystemdwcm{\cal E}$ reduces to $\onlinesystemdpwb{\cal E}$.)
\end{enumerate}
\end{theorem}

The proof is made available through 
the technical-report version of this paper~\cite{hem-hem-rot:t:online-bribery}.

It is interesting to note that, assuming $\p \neq \np$, 
bribery does not reduce to corresponding online bribery,
not even for natural systems. For example, 
Approval-Bribery is NP-complete~\cite[Theorem 4.2]{fal-hem-hem:j:bribery},
but we will show below in Theorem~\ref{t:approval} that
$\onlinesystemb{Approval}$ (and
even $\onlinesystemwb{Approval}$ and $\onlinesystempb{Approval}$)
are in P.

We end this section with a simple observation about
unpriced, unweighted online bribery.
\begin{observation}\label{o:bribers-last}
For unpriced, unweighted online bribery, it is always optimal to bribe
the last $k$ voters (we don't even have to handle $u$ in a special
way). This implies that unpriced, unweighted online bribery
is certainly reducible to unweighted online
manipulation, and so we inherit those upper bounds.
\end{observation}

\subsection{Plurality}

In this section, we completely classify the complexity of all
our versions of online bribery for the most important natural 
system, Plurality. In this system, each candidate scores a
point when it is ranked first in a vote and the candidates
with the most points are the winners.  We show that these
problems are NP-complete if we have both prices and weights,
and in P in all other cases. 

In the constructive
case we call all members of $\{c \ | \ c \geq_\sigma d\}$
and in the destructive case  we call
all members of $\{c \ | \ c >_\sigma d\}$
\emph{desired} candidates, where $\sigma$ 
  is the briber's ideal ranking and $d$ the designated candidate.
The following observation is crucial in our upper
bound proofs:
For $\onlinesystempwb{Plurality}$ and
$\onlinesystemdpwb{Plurality}$,
there is a successful bribery if
and only if there is a successful bribery where all
bribed voters
from $u$ onward vote for the same
highest-scoring desired candidate
and all nonbribed voters
after $u$ vote for
the same highest-scoring undesired candidate.
If $u$ is bribed, we
do not count $u$'s original vote to compute the highest score. If
$u$ is not bribed, then we count $u$'s vote.

\begin{theorem}\label{t:plurality-easy}
$\onlinesystemb{Plurality}$,
$\onlinesystemdb{Plurality}$,
$\onlinesystemwb{Plurality}$,
$\onlinesystemdwb{Plurality}$,
$\onlinesystempb{Plurality}$, and
$\onlinesystemdpb{Plurality}$
are in \p.
\end{theorem}

\begin{theorem}
\label{t:plurality-hard}
$\onlinesystempwb{Plurality}$ and
$\onlinesystemdpwb{Plurality}$ are \np-complete,
even when restricted to two candidates.
\end{theorem}

The proofs of the above two theorems are made available through 
the technical-report version of this paper~\cite{hem-hem-rot:t:online-bribery}.

\subsection{Beyond Plurality}

A scoring rule for $m$ candidates is
a vector  $\alpha = (\alpha_1, \ldots, \alpha_m)$ of integers
$\alpha_1 \geq \alpha_2 \geq \cdots\geq\alpha_m \geq 0$.
This defines an election system on $m$ candidates
where each candidate earns $\alpha_i$ points for each vote
that ranks it in the $i$th position and
the candidates with the most points are the winners. 

\begin{theorem}\label{t:scoring-dichotomy}
For each scoring vector $\alpha = (\alpha_1, \ldots, \alpha_m)$,
\begin{enumerate}
\item
$\onlinesystempwb{\alpha}$
and $\onlinesystemdpwb{\alpha}$ are in \p\ if $\alpha_1 = \alpha_m$ and
$\np$-hard otherwise;
\item
$\onlinesystemwb{\alpha}$ and
$\onlinesystemdwb{\alpha}$ 
are in \p\ if $\alpha_2 = \alpha_m$ and 
$\np$-hard otherwise; and 
\item
$\onlinesystemb{\alpha}, \onlinesystemdb{\alpha},
\onlinesystempb{\alpha}$, and $\onlinesystemdpb{\alpha}$
are in \p.
\end{enumerate}
\end{theorem}

The proof is made available through 
the technical-report version of this paper~\cite{hem-hem-rot:t:online-bribery}.

In Veto, each candidate scores a point when it is not
ranked last in a vote and the candidates
with the most points are the winners.
Now let's look at 3-candidate-Veto.

\begin{theorem}
  \begin{enumerate}
\item
$\onlinesystemb{\mbox{\rm 3-candidate-Veto}}$,
$\onlinesystemdb{\mbox{\rm 3-candidate-Veto}}$,
$\onlinesystempb{\mbox{\rm 3-candidate-Veto}}$, and
$\onlinesystemdpb{\mbox{\rm 3-candidate-Veto}}$ are in \p.
\item
$\onlinesystemwb{\mbox{\rm 3-candidate-Veto}}$ and
$\onlinesystemdwb{\mbox{\rm 3-candidate-Veto}}$ are
$\p^{\np[1]}$-complete.
\item
$\onlinesystempwb{\mbox{\rm 3-candidate-Veto}}$
and
$\onlinesystemdpwb{\mbox{\rm 3-candidate-Veto}}$
are $\p^{\np[1]}$-hard and
in $\deltatwo$ (and we conjecture that they are
$\deltatwo$-complete).
\end{enumerate}
\end{theorem}

\begin{proof}
We look at different cases
for the placement of the designated candidate in the preference order
$a >_\sigma b >_\sigma c$.
(Note that the polynomial-time cases also follow from 
Theorem~\ref{t:scoring-dichotomy}.)

\begin{enumerate}
\item
If $d = c$ in the constructive case
or $d = a$ in the destructive case, the problem is trivial.
\item
If $d = a$ in the constructive case, we 
need to ensure that $a$ is a winner, and
if $d = b$ in the destructive case, 
we need to ensure that $a$ is the unique winner.
In both these cases, the nonbribed voters veto $a$, no matter what.
This means that the location of the bribed voters doesn't matter (though
their prices and weights will). In the priced case,
bribe the
cheapest voters (as many as possible within the budget)
and bribe to minimize the maximum score of $b$ and $c$.  So that's in P.
The priced and weighted (and so also the weighted) case is in NP:
Guess a set of voters to bribe, check that they are within the budget, 
guess votes for the bribed voters, and
have all nonbribed voters veto $a$. NP-hardness for the weighted case
follows from the NP-hardness for weighted manipulation plus
the proof of Theorem~\ref{t:m-to-b}.
\item
If $d = b$ in the constructive case, the goal is to not have $c$
win uniquely, and if $d = c$ in the destructive case,
the goal is to have $c$ not win.
In this case, all bribed voters veto $c$ (no matter what).
The same argument as above shows that priced bribery is in P.

The unpriced, weighted cases are coNP-complete. To show that the
complement is in NP, 
pick the $k$ heaviest voters to bribe. Then check if you can
partition the remaining voters to veto $a$ or $b$ in such a way
that $c$ wins uniquely in the constructive case or
that $c$ wins in the destructive case.
Note that it is always best for the
briber to bribe the $k$ heaviest voters: Swapping the weights of 
a lighter voter to be bribed with a heavier voter not to be bribed
will never make things worse for the briber.
To show hardness, note that the complement is
basically
(the standard NP-complete problem)~Partition.

To show the upper bound for the priced, weighted case, note that we need
to check that there exists a set of voters that can be bribed within
the budget such that if all bribed voters veto $c$, then for all
votes for the nonbribed voters, $c$ is not the unique winner (in the
constructive case) or not a winner (in the destructive case).
This is clearly in $\sigmatwo$. With some care, we can show that it is
in fact in $\deltatwo$. 
First use an NP oracle to determine the largest possible
total weight (within the budget) of bribed voters.
Then determine whether we should bribe $u$, by using the oracle
again to determine the largest possible total weight (within the budget)
of bribed voters, assuming
we bribe $u$. If that weight is the same as the previous weight,
bribe $u$. Otherwise, do not bribe $u$. Repeating this will give us a
set of voters to bribe of maximum weight. All these voters will 
veto $c$. It remains to check that for all votes for the nonbribed voters,
$c$ is not the unique winner (in the
constructive case) or not a winner (in the destructive case).
This takes one more query to an NP oracle.
\end{enumerate}

Putting all cases together, 
the priced case (and so also the unpriced, unweighted case) is in P.
The weighted case is $\p^{\np[1]}$-complete,
since it can be written as the union of a NP-complete set and
a coNP-complete set that are P-separable.
(Two sets $A$ and $B$
are P-separable~\cite{gro-sel:j:complexity-measures} if there exists
a set $X$ computable in polynomial time
  such that $A \subseteq X \subseteq \overline{B}$.)

The priced, weighted case inherits the  $\p^{\np[1]}$-hardness
from the weighted case (in fact, it already inherited this from
online manipulation, using Theorem~\ref{t:m-to-b}), and it is in
$\deltatwo$.~\end{proof}

We end this section by looking at
approval voting.
In approval voting, each candidate scores a point when it is approved
in a vote and the candidates
with the most points are the winners.
Though Approval-Bribery is
NP-complete~\cite[Theorem 4.2]{fal-hem-hem:j:bribery},
we show that $\onlinesystemb{Approval}$ (and
even $\onlinesystemwb{Approval}$ and $\onlinesystempb{Approval}$)
are in P. This implies that even for natural systems, 
bribery can be harder than online bribery (assuming $\p \neq \np$).

\begin{theorem}\label{t:approval}
  \begin{enumerate}
\item
$\onlinesystemb{Approval}$,
$\onlinesystemdb{Approval}$,
$\onlinesystempb{Approval}$,
$\onlinesystemdpb{Approval}$,
$\onlinesystemwb{Approval}$, and 
$\onlinesystemdwb{Approval}$
are each in $\p$.
\item
$\onlinesystempwb{Approval}$ and
$\onlinesystemdpwb{Approval}$ are
each $\np$-complete.
\end{enumerate}
\end{theorem}

The proof is made available through 
the technical-report version of this paper~\cite{hem-hem-rot:t:online-bribery}.

\section*{Acknowledgments}
We are grateful to the anonymous TARK 2019 referees for helpful comments.
\bibliographystyle{eptcsalpha}
%
%
\newcommand{\etalchar}[1]{$^{#1}$}

\end{document}